\documentclass[12pt,titlepage]{amsart}
\usepackage{amsmath,amssymb,amsthm,verbatim,xcolor}
\usepackage{tikz}
\usepackage{pgfplots}
\pgfplotsset{compat=1.11}

\usepackage[foot]{amsaddr}

\usepackage{natbib}

	\usepackage[left=1.4in,top=1.25in,right=1.4in,bottom=1.25in]{geometry}
	\usepackage[onehalfspacing]{setspace}

\def\income{w}

\def\la{\lambda}
\def\al{\alpha}
\def\da{\delta}
\def\w{\omega}\def\W{\Omega}

\def\one{\mathbf{1}}

\def\Na{\mathbf{N}}
\def\Re{\mathbf{R}}

\newcommand{\df}[1]{\textit{\textbf{#1}}}
\newcommand{\abs}[1]{ | #1 | }

\def\P{\mathcal{P}}

	\theoremstyle{plain}
	\newtheorem{theorem}{Theorem}

	\newtheorem*{axiom*}{Axiom}
	
	\newtheorem{proposition}{Proposition}
	\newtheorem{lemma}{Lemma}

	\theoremstyle{definition}

	\theoremstyle{remark}
	\newtheorem*{remark}{Remark}
	
	\newtheorem*{claim*}{Claim}

\begin{document}
\title[Decreasing impatience]{Decreasing impatience}

\author[Chambers]{Christopher P. Chambers}
\author[Echenique]{Federico Echenique}
\author[Miller]{Alan D. Miller}

\address[Chambers]{Department of Economics, Georgetown University}
\address[Echenique]{Department of Economics, UC Berkeley}
\address[Miller]{Faculty of Law, Western University}

\begin{abstract}
We characterize decreasing impatience, a common behavioral phenomenon in intertemporal choice. Discount factors that display decreasing impatience are characterized through a convexity axiom for investments at fixed interest rates. Then we show that they are equivalent to a geometric average of generalized quasi-hyperbolic discount rates. Finally, they emerge through parimutuel preference aggregation of exponential discount factors. \end{abstract}

\thanks{Echenique thanks the National Science Foundation for financial support (Grants SES-1558757 and CNS-518941). This research was undertaken, in part, thanks to funding from the Canada Research Chairs program.  We are grateful to Chris Shannon, Leslie Marx, and four anonymous referees for useful comments and suggestions.}

\maketitle

\section{Introduction}

Decreasing impatience, a property of intertemporal preferences, has been the intense focus of positive studies in behavioral economics, and of a normative literature (and practice) in project evaluation. Following \citet{prelec2004decreasing} \citep[see also][]{kimr:2019}, if an individual is willing to wait a certain number of periods for a higher payoff, that individual would be willing to wait the same number of periods in the future. For example, if an agent who satisfies decreasing impatience prefers to receive a check for \$110 in a week than to receive a check for \$100 in six days, that agent must also prefer to receive a check for \$ 110 in 31 days to \$ 100 in 30 days. We investigate this definition in a linear context. 

Choices that reflect decreasing impatience are usually modeled by means of parametric models of hyperbolic, or $\beta$-$\delta$ (also called quasi-hyperbolic), discounting. See \cite{loewenstein1992anomalies}, or \cite*{frederick2002time} for a review of the experimental literature. Many behavioral economists have used these parametric models of discounting in studies of consumption, savings, and retirement \citep[for example][]{laibson1997golden,o1999doing,diamond2003quasi}. Our paper provides new characterizations of decreasing impatience as a ``non parametric'' property of discount factors.

Our contribution is to provide characterizations of decreasing impatience as a general property of discounting. The property of decreasing impatience is easy to understand, but we think that it is useful to describe it in terms of either more basic behavioral patterns (axioms), or in terms of a ``story'' for how decreasing impatience emerges from other familiar models. 

Our main result proposes three different characterizations. One is axiomatic: within models of intertemporal choice that rely on discounting, a simple convexity axiom captures decreasing impatience. The axiom relies on a single choice problem that can be implemented in the lab; thus avoiding  incentive-compatibility issues that arise with multiple incentivized  decisions \citep{azrieli2018incentives}. In other words, if one wishes to test the property of decreasing impatience, our axiom describes a single question that can be asked of laboratory subjects. An (incentivized) negative answer falsifies the property of decreasing impatience.

Our second characterization is a story about aggregation, mathematically analogous in a sense to \citet{harsanyi1955cardinal}. We show that any discount factor that exhibits decreasing impatience is the result of a geometric average of generalized $\beta$-$\delta$ preferences. In an ancillary result, we provide a justification for using the geometric average by means of a time consistency axiom, together with other standard normative axioms.\footnote{This justification of geometric averaging of discount factors is, we believe, of some interest independently of the present application.} The justification is needed, we believe, to motivate the use of the geometric average instead of other possible aggregation rules. Now, standard $\beta$-$\delta$ preferences treat the present period as special, and worthy  of a premium in intertemporal tradeoffs. Generalized  $\beta$-$\delta$ preferences simply extend the special treatment to all periods before some cutoff date. Our characterization means that $\beta$-$\delta$ discounting is the canonical class of decreasing impatience discount factors: \textit{All other discount factors with decreasing impatience can be understood as a stationary aggregation of such generalized $\beta$-$\delta$ discount factors.} 

Our third characterization shows that decreasing impatience emerges as a property of equilibrium prices in a competitive economy. In fact, we show that decreasing impatience is the \textit{defining} property of  equilibrium prices in a linear dynamic endowment economy:  any discount factor that exhibits decreasing impatience is the equilibrium price of a parimutuel market, where all agents participating in the market are exponential discounters. Parimutuel markets were first proposed as an information-aggregation mechanism by \cite{eisenberg1959consensus}, and have been the focus of an extensive empirical and experimental literature (see, for example, \cite{plott2003parimutuel}). Aside from its traditional use in horse races, they have been implemented inside large corporations as an information aggregation mechanism \citep{gillen2017pari}. Here we use them as preference aggregation mechanisms. Our result means that decreasing impatience always has a representation as the aggregate of a collection of agents with traditional exponential discount factors, where the aggregation takes the pari-mutuel form. Note that instead of a representation in terms of generalized $\beta$-$\delta$ preferences, we obtain any decreasing impatience discount factor by means of aggregating exponential (i.e $\beta=1$) discount factors.

\paragraph{\textbf{Related literature:}} 
Decreasing impatience is a well-known behavioral phenomenon: see for example \cite{loewenstein1992anomalies}, \cite{prelec2004decreasing}, \cite{frederick2002time}, and \cite{kimr:2019}. It is commonly modeled using hyperbolic or quasi-hyperbolic discount factors, and has been incorporated in multiple theoretical studies \citep{laibson1997golden,o1999doing,diamond2003quasi,brw:2009,kimr:2009,kimr:2010,bgr:2016}. 
\cite{halevy2015time} disentangles three related properties: stationarity, time consistency, and time invariance, and shows that any two of these properties imply the third. \cite{chakraborty2016present} considers the phenomenon of present bias in isolation. His weak present bias axiom is satisfied by a host of models that have been introduced to relax the stationarity assumption of exponential discounting. \citeauthor{chakraborty2016present} characterizes the utility representation (within a certain family) that satisfies the axiom of weak present bias.

Some papers provide microfoundations for decreasing impatience.\footnote{The axiom we use to characterize decreasing impatience might be called a form of risk-seeking for time lotteries in \cite{dejarnette2020time}, but the similarity is simply mathematical---our model is fully deterministic, and no lotteries are considered.}   \cite{sozou1998hyperbolic} derives it as a consequence of uncertainty in the discount rate, by taking linear combinations.  In a model in which payoff times of goods are uncertain, \cite{dasgupta2005uncertainty} derive a result about preference reversals through time which reflect a type of decreasing impatience. In a model of exponential discounting,  \cite{halevy2008,saito2011,chakraborty2020relation} relate decreasing impatience to decision-theoretic phenomena occurring in the study of risk.  More recently, \cite{harstad2020technology} shows that the behavior of policy makers may exhibit decreasing impatience due to a time inconsistency problem resulting from the uncertainty as to whether they will be in office from one period to the next.

The problem of aggregating discount rates has received a lot of attention.\footnote{Our result in Section~\ref{sec:multiplicativeaggregators} connects with the literature on  multiplicative aggregation. Our result is  probably most similar to \citet{hayashi2016consistent}, but there are several differences.  The first, and most obvious difference, is that our framework involves no social disagreement over period rewards.  All disagreement is due to the form of discounting.  A second main difference is that we envision this result as being most relevant when applied to dated rewards.  One of our main properties, indeed, is a Pareto condition applied to dated rewards.  Were we to postulate a form of Pareto for streams, we would be back to the framework of \citet{harsanyi1955cardinal} and \citet{jackson2015collective}.  So, intertemporal tradeoffs for consumption streams should be viewed as ``irrelevant'' here. Finally, the point of \citet{hayashi2016consistent} is that, while a form of dynamic consistency may be interesting, we should not necessarily invoke it in an environment in which social preference is independent of history.  By contrast, our framework has no language for allowing us to condition a ranking on history.  So we implicitly rule out his aggregation functions.  The interesting examples motivating his study involve intertemporal tradeoffs across individuals, a phenomenon that does not obtain here.} The seminal paper by \cite{weitzman2001gamma} documents disagreements about the discount rate, and proposes a solution that implies decreasing impatience. Again, our results speak to this literature under additional assumptions that translate discounting of utils to discounting of consumption streams (for example that all agents share a linear utility over consumption). More recently, \cite{zuber2011aggregation} and \cite{jackson2015collective} show that linear aggregation of exponential discounting preferences and time consistency are incompatible.\footnote{Linear aggregation is discussed in Section~\ref{sec:discussion}.}  \cite{feng2018social} and \cite{hayashi2019social} discuss ways of avoiding this impossibility by weakening the assumed Pareto criteria.  \cite{chambers2018multiple} and \cite{chambers2020pareto} introduce and axiomatize decision criteria for environments with multiple discount rates. These papers focus on desirable properties of the resulting aggregate criterion for making intertemporal choices.  In a general context (allowing for nonlinear period utility functions), \cite{millnernondogmatic2020} assumes a kind of multiple-selves model, where at each period each agent has a discount rate, but entertains the possibility that in the next period she will be convinced of a different discount rate. He establishes agreement on discount rates in the long-term.  In particular, he shows that when each individual ascribes a small probability of having an incorrect discount rate, then the long run rate will tend to look like the ``smallest'' discount rate each agent would have had they not admitted the possibility of error.

As a mathematical result, our Theorem~\ref{thm:aggregatorgeomean} is not particularly novel, and indeed was motivated by the log-opinion pool of statistics.  This is a method of aggregating Bayesian priors, by taking a geometric mean of the density functions.  Versions of this aggregator were characterized by \citet{genest1984characterization} and \citet{west1984bayesian} using axioms very similar to the ones we describe here.  It is worth noting that the failure of the log-opinion pool in probability aggregation to commute with respect to marginal distributions (a property used by \citet{mcconway1981marginalization} to characterize linear aggregation) in the case of probability aggregation does not pose a problem for us.  Discount factors over finer or shorter lengths of time are not additive, but multiplicative by their very nature.

\section{The Model}\label{subsec:modelandmotivation}

\subsection{Notational conventions}

A preference relation over a set is a complete and transitive binary relation, also called a weak order. A function $f:A\subseteq\Re\to\Re$ is \df{weakly monotone increasing}, or \df{non-decreasing}, if $f(x)\geq f(y)$ when $x\geq y$; and \df{strictly monotone increasing}, if $f(x)> f(y)$ when $x> y$. It is \df{weakly monotone decreasing}, or \df{non-increasing}, if $-f$ is weakly monotone increasing; and  \df{strictly monotone decreasing} if $-f$ is strictly monotone decreasing.

The set of bounded real sequences is denoted by $\ell^\infty$, and the subset of non-negative sequences by $\ell^\infty_+$.  A sequence $\{x_t\}\in\ell^\infty$ is (absolutely) \df{summable} if $\sum_{t=0}^\infty \abs{x_t}$ converges. The set of summable real sequences is denoted by $\ell^1$, and the subset of nonnegative summable sequences by $\ell^1_+$.  

\subsection{Discount factors}\label{sec:model}

We consider a model of intertemporal choice in which time is discrete, the horizon is infinite, and the objects of choice are bounded real sequences: $\{x_t: t=0,1,\ldots, \}\in\ell^\infty_+$. One may interpret each $x_t$ as a monetary payoff, or as  the value in ``utils'' of some underlying physical outcome. It is worth emphasizing that all our results hold if we assume a finite, instead of an infinite, time horizon.

We restrict attention to preferences that are represented by means of a monotone weakly decreasing \df{discount factor} $f:\Na\to \Re_+$. So a sequence $x$ is ranked above $y$ for the discount factor $f$ if $\sum_{t=0}^\infty f(t) x_t\geq \sum_{t=0}^\infty f(t) y_t$. In fact we shall take $f$ to be a summable sequence and have values in $(0,1]$.  

Of course, these assumptions are not without loss. We restrict attention to preferences with a linear utility representation, $x\mapsto \sum_t x_t f(t)$.  The linear representation presumes a form of independence, or separability; but these assumptions are well understood and merit no further discussion here.\footnote{We could apply certain techniques in the literature, such as the use of paying in lottery tickets and assuming expected utility \citep{rothmalouf} though perhaps our work is better understood as applying when utils are known.} Future payoffs are discounted, as the values of $f$ are weakly decreasing. Moreover the assumption that  discount factors are summable expresses a particular form of impatience (it implies that no weight is placed ``at infinity.'').  The linear structure encapsulates the idea that our sequences represent utils, as the marginal rate of intertemporal substitution depends only on dates, and not on some physical measure of consumption. 

Formally, then, the objects of choice are bounded non-negative sequences: elements of $\ell^\infty_+$.  We consider preferences $\succeq$ on $\ell^\infty_+$ for which there is a monotone weakly decreasing $f\in\ell^1_+$ with the property that, for any $x,y\in\ell^{\infty}$, $x\succeq y$ if and only if $\sum_{t=0}^\infty f(t) x_t\geq \sum_{t=0}^\infty f(t) y_t$. The class of such preferences is denoted by~$\P$.\footnote{We are, of course, not the first to start from these primitives. A salient example is \cite{loewenstein1992anomalies}, who advocate for $\P$ as a model of intertemporal choice, and then impose a version of decreasing impatience that they show implies hyperbolic discounting.}

A \df{dated reward} is a sequence that is identically zero, except for at most one value $t$. Dated rewards are thus identified with pairs $(x,t)$, denoting a sequence that is zero everywhere and equal to $x\geq 0$ at time $t$. Let $\mathcal{D}$ be the set of all dated rewards.  If $\succeq\in\P$ is one of the preferences under consideration, we have that $(x,t)\succeq (y,s)$ if and only if  $x f(t)\geq yf(s)$. So we can say that an agent with preferences $\succeq$ is happy to delay consumption of $y$ at period $t$ in exchange for $x > y$ at $t+1$ if and only if  $\frac{y}{x}< \frac{f(t+1)}{f(t)}$. In particular, the agent is indifferent between consuming or delaying when $\frac{y}{x}= \frac{f(t+1)}{f(t)}$: so the ratio $f(t+1)/f(t)$ is an expression of \df{how impatient} the agent is when it comes to consumption in periods $t$ and $t+1$. It is how much of a discount the earlier payoff has to be, relative to the later higher payoff, for the two to provide the same utility.

The main focus of our paper are preferences for which the ratio $f(t+1)/f(t)$  is monotone weakly increasing. Such preferences, and their associated discount factors, are said to satisfy \df{decreasing impatience}.\footnote{The property may be more properly referred to as ``weak decreasing impatience''. We use the shorthand term ``decreasing impatience'' for simplicity, and do not refer to the strict version in this paper.} 

A preference $\succeq\in \P$ with associated discount factor $f$ is \df{stationary} if the ratio $f(t+1)/f(t)$ is constant; independent of $t$. It is well known, and easy to see, that this case corresponds to the existence of $\da\in (0,1]$ and a scalar $A$ for which $f(t)=A\da^t$ (indeed, $A=f(0)$). The discount factor is then associated with a constant \df{exponential discount rate} $\da$. Note that stationary preferences also display decreasing impatience. The subclass of stationary preferences, also called exponential discounting preferences, is denoted by $\P^S$. 

In our discussion, stationarity and decreasing impatience are defined as properties of $f$. They can also be defined as properties of $\succeq$: \df{decreasing impatience} says that if $x>y$, $s<t$, and $(x,t)\succeq (y,s)$ then $(x,t+r)\succeq (y,s+r)$ for all $r>0$. Stationarity strengthens this to be an ``if and only if'' statement, holding for all $x,y$ (see \cite{chakraborty2016present} for an eloquent discussion of these properties). It is easy to see that the properties of $\succeq$ are equivalent to our definitions, within the class $\P$, see \emph{e.g.}\ \citet{prelec2004decreasing}, Corollary 1, where it first appears.

In applications, it is common to model decreasing impatience through a $\beta$-$\delta$, or quasi-hyperbolic, discount factor \citep{laibson1997golden}: these take the form $f(t)=\beta^{\min\{t,1\}}\da^t$, with $\da,\beta\in (0,1]$, so that $f(t+1)/f(t)$ goes from $\beta \da$ when $t=0$ to $\da$ for all $t>0$. The idea is that period $t=0$ plays a special role. We are interested in a generalization of this model that extends this special role to all initial periods: $t=0,\ldots,t^*-1$ for some $t^*\geq 1$. 

Specifically, say that 
a discount factor $f$ is \df{generalized $\beta$-$\delta$} if there are $\beta,\delta\in (0,1]$ and $t^*$ such that \[ 
f(t) = \beta^{\min\{t^*,t \}} \da^t.
\] In a generalized  $\beta$-$\delta$ discount factor, the measure of impatience $f(t+1)/f(t)$ goes from $\beta \da$ in periods $t=0,\ldots t^*$ to $\da$ in periods $t> t^*$. The standard quasi-hyperbolic model obtains when $t^*=1$, and stationary (exponential) discounting when $\beta=1$.\footnote{A reason for this generalization is that individuals may not measure time periods in the same way. This is easiest to see in a continuous model, where a period corresponds to a length of time before preferences `change', and where there is no reason to think that the length of time should be the same for all agents. Generalized $\beta-\delta$ preferences are simply a way to represent this in discrete time.}

\section{Main results}\label{sec:main}

Before we state our main results, we introduce a few preliminary ideas. The first is a behavioral axiom: a pattern of intertemporal choice which says that for any principal $k$ and rate of return $r$, investing half of $k$ at maturity $t-1$, and half at maturity $t+1$, is always preferred to investing all of $k$ at maturity $t$. Such a pattern of choice is called \df{Compound-interest convexity.} Formally, the statement of the axiom is:

\begin{axiom*}[Compound-interest convexity] For all $k>0$, all $t \geq 1$ and all $r>0$,
\[ 
\left(\frac{k}{2}(1+r)^{t-1},t-1\right)  + \left(\frac{k}{2}(1+r)^{t+1},t+1\right) \succeq  (k(1+r)^{t},t).
\] \end{axiom*}
Recall that the dated reward notation $(x,t)$ refers to a sequence in $\ell^\infty_+$, making the addition of dated rewards meaningful. The expression $\left(\frac{k}{2}(1+r)^{t-1},t-1\right)  + \left(\frac{k}{2}(1+r)^{t+1},t+1\right)$ refers to the sequence that is zero everywhere except for in periods $t-1$ and $t+1$, at which it equals, respectively $\frac{k}{2}(1+r)^{t-1}$ and $\frac{k}{2}(1+r)^{t+1}$. As we shall see, within the class $\P$, compound-interest convexity characterizes decreasing impatience.\footnote{Recall the characterization of decreasing impatience described by \citet{prelec2004decreasing} by log-convexity:  $f(t+1)^2 \leq f(t)f(t+2)$.  A mathematical result, essentially due to \citet{montel1928fonctions}, establishes that this condition is satisfied if and only if for every $\beta >1$, $2\beta^{t+1}f(t+1)\leq \beta^tf(t)+\beta^{t+2}f(t+1)$.  This latter (additively stated) characterization forms the basis of compound-interest convexity.} 

The other two notions we shall introduce are related to aggregating discount factors. The point will be that a discount factor satisfies decreasing impatience if and only if it is the aggregate of some basic parametric models of discounting. 

The first method of aggregation is the geometric mean. Given a finite or countable collection of discount factors $f_s$, a \df{geometric mean} is $\prod_s f_s(t)^{\eta_s}$, for some $\eta_s>0$ with $\sum_s \eta_s=1$. Importantly, in Section~\ref{sec:multiplicativeaggregators} we show that the geometric mean of a finite number of discount factors is the unique aggregation rule that uniquely satisfies a notion of time consistency, together with some standard normative axioms.

A second, perhaps unexpected, connection to decreasing impatience comes from the method of ``parimutuel aggregation'' introduced by \cite{eisenberg1959consensus}.\footnote{The idea in this paper is generalized in \cite{eisenberg1961aggregation}.} The idea is to use a market mechanism (or a pseudomarket, where agents use exogenously given incomes to purchase goods) and have the discount factor arise as an equilibrium price.

A \df{parimutuel economy}  is a collection $(\succeq_i,\income_i)_{i\in I}$, where $I$ is finite or countable,  each $\succeq_i$ being a preference relation in $\P^S$ (meaning a stationary preference over streams in $\ell^\infty_+$), and $\income_i>0$  satisfying that $\sum_{i\in I} \income_i$ is finite. In words, a parimutuel economy consists of a set $I$ of agents with exponential preferences, and strictly positive income $\income_i$, such that aggregate income $\sum_i \income_i$ is well defined. 

We restrict attention to parimutuel economies with a unit supply of ``good,'' or money, per period. An \df{allocation} in a parimutuel economy is a collection $x=(x_i)_{i\in I}$ of sequences in $\ell^\infty_+$ with the property that \[ 
\sum_{i\in I} x_i(t)=1
\] for all $t$.\footnote{Our results do not depend on the assumption that $\sum_{i\in I}x_i(t) = 1$, but rather that $\sum_{i\in I}x_i(t) = \bar\w_t$, for some fixed supply $\w_t>0$,  and that all relevant consumption streams are in the domain.  We maintain the assumption of unit supply for analytical convenience.}

A \df{parimutuel equilibrium} in $(\succeq_i,\income_i)_{i\in I}$ is a pair $(p^*,x^*)$ in which $x^*$ is an allocation and  $p^*\in\ell^1_+$ is a sequence of prices, for which $x^*_i$ is maximal for preference $\succeq_i$ in the budget set 
\[ \{x\in\ell^\infty_+:\sum_t p(t) x(t) \leq \income_i \}.
\]

Finally, for any two sequences $f$ and $g$, $f\propto g$ means that they are proportional to each other. So there is some $\alpha > 0$ for which $f = \alpha g$.  Observe that if $f$ and $g$ are discount factors, then $f \propto g$ means that they represent the same $\succeq$. In particular, our main results hold if we normalize discount factors so that $f(0)=1$; perhaps a natural normalization.

\subsection{Characterization of decreasing impatience}

\begin{theorem}\label{thm:main1}
Let $\succeq$ be a preference in $\P$, with associated discount factor $f$, and suppose that $f(t+1)/f(t)$ is bounded away from 1. 
The following statements are equivalent:
\begin{enumerate}
    \item\label{thm1it1} $\succeq$ satisfies decreasing impatience.
    \item\label{thm1it2} $\succeq$ satisfies compound-interest convexity.
    \item \label{thmit2.5} $f$ is proportional to the (finite or countable) geometric mean of generalized $\beta$-$\delta$ discount factors. That is, there are $\beta$-$\delta$ discount factors $f_s$, and $\eta_s>0$ with $\sum_s \eta_s=1$, such that  $f(t)\propto \prod_s f_s(t)^{\eta_s}$.
    \item\label{thm1it3} There exists a parimutuel economy, and a parimutuel equilibrium $(p^*,x^*)$ in this economy, for which $p^*_t=f(t)$ for all $t$. 
\end{enumerate}
\end{theorem}

\begin{remark}Observe that in order to falsify compound-interest convexity, it is sufficient to find a single observation \[ 
\left(\frac{k}{2}(1+r)^{t-1},t-1\right)  + \left(\frac{k}{2}(1+r)^{t+1},t+1\right) \prec  (k(1+r)^{t},t).
\] The significance of this observation is that, in a framework of dated rewards, decreasing impatience requires at least two observations to falsify.  Multiple observations of choices in experimental economics usually require resorting to some type of random problem selection, see \emph{e.g.} \citet{azrieli2018incentives}, and thus committing to a theory of behavior over random outcomes. By having an axiom that can be falsified with a single observation, there is no need to add assumptions about how experimental subjects treat random outcomes. The tradeoff is that we must commit to a theory over consumption streams for compound-interest convexity to be meaningful.
\end{remark}

\begin{remark} The $\beta$-$\delta$ model is popular in behavioral economics as a tractable approximation to the hyperbolic discount factor (see for example \cite{laibson1997golden} or \cite{diamond2003quasi}).\footnote{See \citet{loewenstein1992anomalies} and the references therein for a discussion of hyperbolic discounting.  A form of the $\beta$-$\delta$ model was axiomatized by \citet{hayashi2003quasi}.} In contrast, the equivalence between~\eqref{thm1it1} and~\eqref{thmit2.5} means that  generalized $\beta$-$\delta$ preferences are, in a sense, the canonical model of decreasing impatience. Given that  we can turn any $\beta$-$\delta$ preference into a generalized $\beta$-$\delta$ preference by suitably redefining the length of the initial time period of each agent, the equivalence in Theorem~\ref{thm:main1} implies that one may always think of a discount factor satisfying decreasing impatience as an  aggregate  of $\beta$-$\delta$ discount factors. Moreover, as emphasized by Theorem~\ref{thm:aggregatorgeomean} below, the geometric mean as an aggregator of discount rates satisfies a notion of time consistency.  (As will become apparent below, this aggregator makes the most sense when restricting to $\mathcal{D}$.)
\end{remark}

\begin{remark} Given the results by \cite{weitzman2001gamma} and \cite{jackson2015collective}, one might wonder if all discount factors with decreasing patience might not be obtained through utilitarian aggregation of exponential discount factors. It is, however, easy to see that standard quasi-hyperbolic discount factor cannot be obtained in this fashion. Indeed if this were the case then there would be a (potentially countable) collection of discount factors $\da_i$, with weights $a_i$ such that $\beta \da^t=\sum_i a_i \da^t_i$ for all $t>1$. Put differently,  $\beta=\sum_i a_i (\da_i/\da)^t$, so the right-hand side is constant in $t$, which is only possible if $\da_i=\da$ for all $i$. See Section~\ref{sec:linearaggregation} for more on what is possible with utilitarian aggregation.
\end{remark}

\begin{remark}An inspection of the proof of Theorem~\ref{thm:main1} establishes that decreasing impatience implies something apparently stronger (but equivalent to) than compound-interest convexity.  In particular, decreasing impatience implies that for all $k>0$, all $t\geq 1$ and all $\beta > 0$, \[\left(\frac{k}{2}\beta^{t-1},t-1\right)  + \left(\frac{k}{2}\beta^{t+1},t+1\right) \succeq  (k\beta^{t},t).\]
This stronger hypothesis provides further simple tests for refuting the hypothesis of decreasing impatience.
\end{remark}

\begin{remark}
The hypothesis that $f(t+1)/f(t)$ is bounded away from 1 guarantees that each of the $\beta$-$\delta$ preferences obtained in~\eqref{thmit2.5} has exponents that are strictly smaller than 1. If we only assume that $f$ is strictly decreasing, the remaining equivalences in the theorem continue to hold.
\end{remark}

\begin{remark}On condition~\eqref{thm1it3}:  in fact, any collection of agents possessing decreasingly impatient discount factors and participating in a parimutuel market will generate prices that also exhibit decreasing impatience in any equilibrium.  Thus, if we only observe prices, we cannot rule out that all agents in the economy are decreasingly impatient.  Our result claims that we also cannot preclude the fact that all agents in the economy are additionally exponential discounters.

Our result follows from the expression for equilibrium prices in linear economies, which says that prices are the upper envelope of agents' (scaled) utility indexes. The fact that equilibrium prices exhibit decreasing impatience then follows because log convex functions are suprema of exponential functions. It is also worth mentioning that, under our assumptions, equilibrium prices are unique.  These issues, and the uniqueness result for this model with infinitely many agents, are discussed in the Appendix in Section~\ref{sec:generalmodel}.\end{remark}

\section{Multiplicative aggregators and dated rewards}\label{sec:multiplicativeaggregators}

In this section, we shed light on Statement~\eqref{thmit2.5} in Theorem~\ref{thm:main1} by providing a foundation for the geometric average of discount factors as a preference-aggregation method in intertemporal choice. We focus in this section on the classical framework of dated rewards, as in \citet{fishburn1982time}.  This is a smaller domain than the domain of consumption streams.  

We discuss a formal model of preference aggregation, and establish the class of multiplicative aggregators as the unique ones satisfying a collection of properties.  The idea here is that utility is common, and consumption is public, but a collection of agents have idiosyncratic preferences over the common utility streams.  That is, agents are asked to rank streams; the only disagreements are about discount factors.  We envision the exercise here as making the most sense for the domain $\mathcal{D}$ of dated rewards. We imagine that the goal is to aggregate a group of individual discount factors into a social one, and impose several properties on how this aggregation takes place.  Key amongst our assumptions are a Pareto property \emph{for dated rewards only} and a time consistency property.

Let us denote by the set of discount factors by \[\mathcal{NI} \equiv \{f: \Na \to \Re_+ : f \text{ is non-increasing and } f(0)=1\}.\]  Here for simplicity we focus on dated rewards, and therefore it is enough to consider non-increasing discount factors.  That $f(0)=1$ reflects a basic normalization.

For any time period $t$ and $f\in \mathcal{NI}$, define $f^t\in \mathcal{NI}$, the \df{$t$-shifted} version of $f$, by $f^t(s) = \frac{f(t+s)}{f(t)}$.  The $t$-shifted version of $f$ is the discount factor that would obtain if the decisions made using $f$ for period $t+s$ would be revisited after $t$ periods have passed. With the definition of $f^t$, after being revisited in period $t$, any decisions would be maintained. 

Given is a finite set of agents $M\equiv \{1,\ldots,m\}$, indexed by $i\in M$.  An \df{aggregator} is a function $\varphi:\mathcal{NI}^M\rightarrow\mathcal{NI}$.  The idea behind an aggregator is that there is a social preference, which takes the same form as individual preference (so it is represented by a discount factor), and is determined as a function of individuals' discount factors.  In general, the aggregate  discount factor evaluated at a particular time $t$ could depend on the entire sequence of discount factors for every individual agent.

An aggregator $\varphi$ is a \df{geometric mean} if there exists  $\eta_i >  0$ for each $i\in M$ such that  $\sum_i \eta_i = 1$  and \[ 
\varphi(f_1,\ldots,f_m)(t)=\prod_{i\in M}f_i(t)^{\eta_i},
\] for all $(f_1,\ldots,f_m)\in {\mathcal{NI}}^M$.\footnote{If we reformulate our model in log terms, so that $f_i(t) = \exp{g_i(t)}$, then the geometric mean of $f_i(t)$ corresponds to the arithmetic mean of $g_i(t)$.} 

We postulate the following axioms:

\begin{itemize}
\item (Pareto) If for all $i\in M$, $f_i(t)x\geq f_i(s)y$, then $\varphi(f_1,\ldots,f_n)(t)x \geq \varphi(f_1,\ldots,f_n)(s)y$, with a strict inequality if any individual inequality is strict.
\item (Independence of Irrelevant Alternatives) For any $t,s>0$.  For all $f,f'\in\mathcal{NI}^M$, if for all $i\in M$ and all $x,y\in\Re_{++}$: $f_i(t)x \geq f_i(s) y$ iff $f'_i(t)x\geq f'_i(s) y$, then for all $x,y\in\Re_{++}$, $\varphi(f)(t)x\geq \varphi(f)(s)y$ iff $\varphi(f')(t)x \geq \varphi(f')(s)y$.
\item (Time consistency) For all $t\geq 0$, $\varphi(f^t_1,\ldots,f^t_m)= \varphi(f_1,\ldots,f_m)^t$.
\end{itemize}

The Pareto axiom says that if all agents agree on the ranking of a dated reward, then this ranking should be respected by the aggregate discount factor. Independence of irrelevant alternative demands that in making an aggregate comparison between dated rewards involving dates $t$ and $s$, only the agents' discount factors involving those two dates should matter. Given the role of our next result in Theorem~\ref{thm:main1}, we want to emphasize the Time consistency axiom. The terminology is inspired by \citet{halevy2015time}.\footnote{We thank the anonymous referees for suggesting this terminology for the property.}

Suppose that we ask individuals about their preference between two dated rewards when $t$ periods have passed, each of them treating time $t$ as if it were the new period $0$. And suppose that the aggregator then judges $(x,s)$ to be preferred to $(x',s')$. Time consistency requires that the aggregator at the original time 0 should judge $(x,s+t)$ to be preferred to $(x',s'+t)$. Otherwise a plan for choosing $(x',s'+t)$ over $(x,s+t)$ would be reversed when time $t$ arrives.\footnote{Strictly speaking, this axiom is an axiom on the aggregator, and is therefore not formally the same as the concept in \citet{halevy2015time}, which does not speak about aggregators.  It essentially requires two things:  first, that the social preference at time $t$ can be derived from applying the time $0$ aggregator to the individual preferences at time $t$.  Second, it requires that this social preference is consistent in the sense of \citet{halevy2015time}.}

The remaining axioms should be familiar. 
Pareto is the usual Pareto efficiency axiom restricted to dated rewards:  were we to apply Pareto to a domain of streams, we would end up in an environment similar to \cite{harsanyi1955cardinal}, whereby aggregation would be additive. Finally 
Independence of Irrelevant Alternatives (IIA) is a version of Arrow's IIA: it says that given a period $t$, the aggregator should only use information about the sets of agents that ranks a date $t$ reward against a period $0$ reward.

\begin{theorem}\label{thm:aggregatorgeomean}An aggregator satisfies Pareto, IIA, and Time consistency if and only if it is a geometric mean. \end{theorem}

\begin{remark}
The geometric mean of summable discount factors may not be summable, but if the discount factors in question are generalized $\beta$-$\delta$ with $\delta<1$ then their geometric mean is guaranteed to be summable. In any case, summability is not needed to rank dated rewards, which has been our focus in this section.
\end{remark}

\begin{remark}
The axioms in Theorem~\ref{thm:aggregatorgeomean} are independent.  Pareto is violated by a  constant $\varphi$ which always returns the sequence $(1,\beta,\beta^2,\ldots)$ for some $0<\beta<1$.  IIA is violated by an aggregator $\varphi$ whose behavior depends on the tails of the  $f_i$.  For example, pick $i,j\in M$, $i\neq j$.  If $\lim_t \frac{f_1(t)}{f_2(t)}=0$, define $\varphi(f_1,\ldots,f_m)(t)=\prod f_i(t)^{\alpha_i}$, otherwise define $\varphi(f_1,\ldots,f_m)(t)=\prod f_i(t)^{\beta_i}$, where $\beta \neq \alpha$.  Finally, time consistency is violated by the rule $\varphi(f_1,\ldots,f_m)(t) = \frac{\sum_i f_i(t)}{m}$.
\end{remark}

\begin{remark}
Observe that Theorem~\ref{thm:aggregatorgeomean} allows arbitrary discount factors, but establishes a unique method of aggregation.  Theorem~\ref{thm:main1} shows that a certain class of functions, the $\beta$-$\delta$ ones, form a kind of ``basis'' of the decreasing impatience discounts for this method of aggregation.
\end{remark}

\section{Discussion and conclusion}\label{sec:discussion}

\subsection{Linear aggregation}\label{sec:linearaggregation}
Many previous studies have focused on linear aggregation of exponential discount rates \citep{zuber2011aggregation,jackson2015collective}. Our Theorem~\ref{thm:main1} provides some alternative representations, but it turns out that it is possible to use related ideas to obtain a linear representation for any discount factor that displays decreasing impatience. That is, a statement analogous to the equivalence between~\eqref{thm1it1} and \eqref{thmit2.5} in the theorem, but with a linear function instead of a multiplicative one. The representation is not, however, in terms of exponential or $\beta$-$\delta$ discount factors. As we have already pointed out, it is in general impossible to obtain linear representation in terms of exponential discount factors. A variation on the argument in Section~\ref{sec:main} shows that it is also impossible to obtain a representation in terms of generalized $\beta$-$\da$ discount factors.

A very basic insight behind Theorem~\ref{thm:main1} is that the set of discount factors that satisfy decreasing impatience is convex, and so can be represented in terms of its extreme elements. Indeed, the proof of Theorem~\ref{thm:main1} reveals that $f$ and $g$ satisfy decreasing impatience if and only if $\beta^t f(t)$ and $\beta^t g(t)$ are convex functions, for $\beta>1$. But then when $\la\in (0,1)$, $\beta^t (\la f(t)+ (1-\la) g(t))$ is convex; establishing the convexity of the set of discount factors with the decreasing impatience property. 

Now,  using results from \cite{langberg1980extreme}, one can show that for each of these extreme elements, there is a (potentially infinite)  increasing sequence of discount factors, $(\beta_1,\beta_2,\ldots)$.  Each discount factor, except possibly the first one, is used for \emph{at least two consecutive periods}, meaning that for each $\beta_l$, there are two periods $t,t+1$ for which $\frac{f(t+2)}{f(t+1)}=\frac{f(t+1)}{f(t)}=\beta_l$.  As these form extreme rays of the relevant class of discount factors, classical Paretian aggregation assuming linearity (as in \citep{harsanyi1955cardinal}) would mean that these discount factors form ``canonical'' ones from which all others can be built linearly. This establishes another kind of representation.  The extremal discount factors figuring in this representation are of course a proper superset of the ones invoked in part \eqref{thmit2.5} of Theorem~\ref{thm:main1}. We have chosen to emphasize the generalized $\beta$-$\delta$ discount factors because of their connection to popular models in behavioral economics.

\cite{langberg1980extreme} is devoted to decreasing failure life rate distributions.  A decreasing failure rate in their paper is determined by log-convexity of the decumulative distribution function.  The authors in that paper characterize the extreme points of the log-convex decumulative distribution functions.  Our paper instead focuses on decreasing sequences of discount factors, but up to scale the mathematics behind the two concepts are identical:  a decreasing nonnegative sequence that satisfies log-convexity.  In the context of \citet{langberg1980extreme}, summability is not a focus, but otherwise the concepts are the same; and a close inspection of their arguments establishes that summability poses no special issue.

\subsection{Sequences of transformations}

Condition~\eqref{thm1it2} of Theorem~\ref{thm:main1} would allow us to provide a characterization of pairs $x,y\in\ell^{\infty}$ for which $x \succeq y$ for every $\succeq\in \P$ (we would additionally need to introduce linear inequalities asserting that discount rates are nonincreasing).  Such a result would claim that $x\succeq y$ for all $\succeq\in \P$ iff $x$ arises from $y$ from a sequence of transformations; analogous to mean-preserving spreads as in \citet{rothschild1970increasing}.  A similar exercise appears in \citet{chambers2020pareto}.

\subsection{Increasing impatience} There is some empirical support for increasing impatience, meaning that there are environments in which some subjects display increasing impatience. It is possible to derive analogous results to ours for this property. In particular, note that the starting point for our main results is the fact that decreasing impatience is equivalent to  the log-convexity of the discount factor. For increasing impatience, one would instead analyze log-concavity. It is then possible to derive results along the lines of the first two characterizations in our Theorem~\ref{thm:main1}.\footnote{We thank Peter Wakker for pointing out us to the relevant empirical literature on increasing impatience.}

\subsection{Conclusion}

Our paper presents a general non-parametric analysis of decreasing impatience as a property of discount factors in intertemporal choice. We have considered its testable implications, in terms of a simple behavioral axiom, and its foundation as an aggregate of more basic parametric models of discounting. 

The results have been developed in the context of discrete time, and taking as given a fixed utility  representation over the underlying physical outcomes. A natural next step is to relax these restrictions, and consider at the same time the problem of aggregating per-period utilities, as well as intertemporal tradeoffs. Another interesting question relates to the family of transformations of a utility stream that preserves preference, for any preference that satisfies decreasing impatience. This would allow for a characterization of all the binary comparisons that any discount factor satisfying decreasing impatience would agree on, along the lines of the exercise in \cite{chambers2020pareto}.

\clearpage

\section{Appendix}\label{sec:proofs}
\subsection{General parimutuel markets}\label{sec:generalmodel}
Equilibria in parimutuel economies may be viewed as solutions to a particular kind of social welfare maximization problem. Indeed, \cite{samuelsonaggregation1956} proposed a general aggregation procedure whereby a representative consumer arises from the maximization of a social welfare functional. With the right prices, these solutions can be decentralized, as in the second welfare theorem, to be consistent with individual optimizing behavior. For parimutuel, or Eisenberg-Gale, aggregation in our context, the  social welfare function in question is the so-called Nash welfare \citep{nash1950bargaining}\[ 
W((u_i)_{i\in I}) = \prod_{i\in I}u_i^{\income_i}, %\sum_{i\in I} \income_i \log u_i,
\] where $u_i(x) = \sum_t x_t \da^t_i$ represents $\succeq_i$. The social welfare maximization program is then \[ 
\begin{array}{cc}
\max_{x_i\in\ell^\infty_+}     & W((u_i)_{i\in I})  \\
\text{s.t}     & \sum_i x_{i,t} = 1 \text{ for all }t.
\end{array}
\] The equilibrium allocations identified in Theorem~\ref{thm:main1} solve this maximization problem, and equilibrium prices take the form of the upper envelope of ``weighted'' versions of the agents discount factors. An illustration is provided in Figure~\ref{fig:parimutueleqprice}.

\begin{figure}[htt]
    \centering
\pgfplotsset{ticks=none}
\begin{tikzpicture}[scale=1.1]
\begin{axis}[%grid=both,
          xmax=10,ymax=1,ymin=0,xmin=0,
          axis x line=bottom,
          xlabel style={color=black},
          xlabel = {\footnotesize $t$},
          axis y line=left,
          color=black!25,
%          ytick={0},
          %ylabel = $f(t)$,
          restrict y to domain=-7:12,
          enlarge x limits]
%          enlargelimits]
\addplot[green!80!black,domain=0:10,very thin,samples=100]  {pow(.3,x)} node[above]{};%{$y=\da_1^t$};
\addplot[blue,domain=0:10,very thin,samples=100]  {0.65*pow(.6,x)} node[above]{};%{$y=\da_2^t$};
\addplot[red,domain=0:10,very thin,samples=100]  {0.3*pow(.8,x)} node[above]{};%{$y=\da_3^t$};
\addplot[black,domain=0:10,samples=100] {max(pow(.3,x),0.65*pow(.6,x),0.3*pow(.8,x))} node[above]{\footnotesize $f(t)$};
%\draw[black!20,ultra thin] (-1,0) -- (10,0); 
%\draw node(-1,0)[left]{\footnotesize $0$};
%\addplot[blue,domain=1/2^6:10,samples=100]  {log2(x)} node[above left] {$y=\log_2(x)$};
\end{axis}
\end{tikzpicture}
    \caption{Parimutuel equilibrium price with three exponential discount factors: {\color{green!80!black}$\da_1$} $<$ {\color{blue}$\da_2$} $<$ {\color{red}$\da_3$}.}
    \label{fig:parimutueleqprice}
\end{figure}
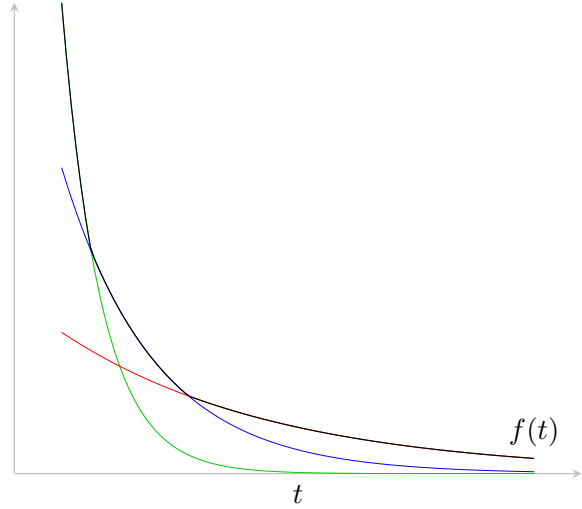

In the figure, there are three agents with exponential discount factors. The equilibrium prices are indicated in black, as the pointwise maximum of the agents' weighted discount factors. It should be clear from the picture that the price exhibits decreasing impatience. Theorem~\ref{thm:main1} says that any discount factor with this property can be interpreted as such an equilibrium price. 

Now, in  Theorem~\ref{thm:main1}, the discount factor was taken as the primitive starting point. In contrast, in this section we take a population of agents $N$ as the starting point. Each of the agents $i\in N$ have preferences in $\P$, and  we consider an aggregate discount factor obtained through parimutuel aggregation.  The next result concerns the structure of the set of possible prices for parimutuel equilibria with a given set of preferences.  We study the set of possible prices as incomes vary.

The result we present is very general, covering environments with both continuous and discrete time. The reason for bothering with this level of generality is that it is usually much easier to compute examples in continuous time, so we want to have a result that can be applied to a continuous time model. On the other hand, the main results of the paper were stated for an environment with discrete time, and we want a result that applies to the same environments as Theorem~\ref{thm:main1}.\footnote{It is, however, not too difficult to obtain a version of Theorem~\ref{thm:main1} for continuous time. At least versions of the equivalence between statements~\eqref{thm1it1}, \eqref{thm1it2} and ~\eqref{thm1it3}.} In the end, it turns out that there is a common structure that works quite generally. 

Let $(\Omega,\Sigma)$ be a measurable space, and for each $i\in N$, let $\da_i$ be countably additive probability measure on $(X,\Sigma)$. We assume that the set $\{\da_i\}_{i\in N}$ is mutually absolutely continuous. The discrete time model is obtained when  $(\Omega,\Sigma)=(\Na,2^{\Na})$ and $\da_i$ is identified with an exponential measure on $2^\Na$ (that is, with a number $\hat \da_i\in (0,1)$ so that $\da_i(A) = \sum_{t\in A}(1-\hat \da_i)\hat \da^t_i$).  The continuous-time model is obtained when  $(\Omega,\Sigma)=(\Re_+,\mathbf{B})$, where $\mathbf{B}$ is the Borel $\sigma$-algebra on $\Re_+$ and $\da_i$ is an exponential probability measure on $\mathbf{B}$.

In a parimutuel market, $\income_i\geq 0$ denotes $i$'s wealth. Here we assume that $\sum_{i\in N} \income_i > 0$.  An \df{economy} then consists of probabilities and wealth.  An \df{allocation} consists of, for each $i\in N$, $x_i \in \L_+^{\infty}(\Omega,\Sigma)$, for which $\sum_{i\in N} x_i = \mathbf{1}$.\footnote{Of course, we could also describe the preferences induced by the probability measures $\da_i$ as we did in the preceding sections.}

A \df{parimutuel equilibrium} is a pair $(p^*,x^*)$, consisting of a finite non-negative measure\footnote{We do not impose countable additivity. In fact, this countable additivity will be shown to be a consequence of equilibrium.} $p^*$ and an allocation $x^*=\{x^*_i\}_{i\in N}$ for which for all $i\in N$:

For all $g\in L_{+}^{\infty}(X,\Sigma)$, $\int g dp^* \leq w_i$ implies $\int g d\da_i \leq \int x^*_i d\da_i$.

Now, for any measure $\da_i$ and any scalar $\al_i$, $\al_i \da_i$ denotes the scalar multiple of the measure.  Then $\bigvee_{i\in N}\al_i \da_i$ denotes the join of the measures in the pointwise dominance order. See, for example, \citet{guide2006infinite}, Theorem 10.56.

In particular, $p = \bigvee_i \al_i \da_i$ exactly when there is a measurable partition $\{E_1,\ldots,E_n\}$ of $\Omega$ for which
\begin{enumerate}
\item For all $E\in\Sigma$ and all $i\in N$, $p(E)\geq \al_i \da_i(E)$.
\item For all $E\in \Sigma$, $p(E) = \sum_{i\in N}\al_i \da_i(E\cap E_i)$.
\end{enumerate}

\begin{proposition}\label{prop:join}Suppose given $\{\da_i\}_{i\in N}$ and $\{\income_i\}_{i\in N}$ for which $\sum_{i\in N} \income_i > 0$.  For any equilibrium $(p^*,x^*)$ of the corresponding economy, there is $\al_i\geq 0$, with $\sum_{i\in N} \al_i >0$ for which $p^* = \bigvee_{i\in N}\al_i \da_i$.  Conversely, if there are $\al_i \geq 0$ for which $\sum_{i\in N} \al_i > 0$ and $p^* = \bigvee_{i\in N}\alpha_i \da_i$, then there for all $i\in N$, there is $\income_i\geq 0$ with $\sum_i \income_i >0$ for which $p^*$ constitutes an equilibrium price in the resulting economy.\end{proposition}

The following establishes uniqueness of Eisenberg-Gale aggregation, supposing mutual absolute continuity of $p_i$.  The proof essentially replicates the argument found in \cite{eisenberg1959consensus}.

\begin{proposition}\label{prop:uniqueness}In the framework of Proposition~\ref{prop:join}, if $(p,x)$ and $(\bar{p},\bar{x})$ are equilibria, where each of $p$ and $\bar{p}$ is a probability measure, then $p = \bar{p}$.\end{proposition}

Of note is that we have not been able to establish existence in general; though it is clear that prices of the form Proposition~\ref{prop:join} are the only potential prices one needs to check.  The literature seems to have been unable to address this problem in our context; relevant works include \cite{wilson1981,richard1988,burke1988}.  These papers must assume that preferences are not strictly monotonic, which does not hold in our context.

\subsection{Proof of Theorem~\ref{thm:main1}}
First note that \eqref{thm1it1} holds iff $\log f(t)$ is (discretely) convex in $t$: 
 log-convexity means that $2\log(f(t+1))\leq \log(f(t))+\log(f(t+2))$, or $\frac{f(t+1)}{f(t)}\leq \frac{f(t+2)}{f(t+1)}$.  We claim that log-convexity is equivalent to the convexity of $\beta^t f(t)$ for any $\beta>1$, a property that is equivalent to Statement~\eqref{thm1it2} in the theorem.\footnote{The equivalence between decreasing impatience and log-convexity is emphasized by \cite{prelec2004decreasing}. See his Corollary 1. The equivalence between log-convexity and the convexity of  $\beta^t f(t)$  is essentially an idea from \cite{montel1928fonctions}.}
 
Convexity of $\beta^t f(t)$ in $t$ means that for every $t\geq 0$, $\beta^t f(t) + \beta^{t+2} f(t+2) \geq 2 \beta^{t+1} f(t+1)$, hence \[ 
h(\beta) =  \beta^2 f(t+2) - 2 \beta f(t+1) + f(t) \geq 0.
\] So fix $t \geq 1$, and observe that $h$ is convex in $\beta$ as $f(t+2)>0$.  We solve for the minimum value of $h$ over $\beta$: the first-order condition gives $2\beta f(t+2) - 2 f(t+1)=0$. Now, since  $\beta = f(t+1)/f(t+2)\geq 1$ (as $f$ is monotone decreasing), the minimum value of  $h$ is 
\[ \left(\frac{f(t+1)}{f(t+2)}\right)^2 f(t+2) - 2  \left(\frac{f(t+1)}{f(t+2)}\right) f(t+1) + f(t) \geq 0.
\] The inequality occurs when $f(t+1)=f(t+2)$ by continuity of $h$ in $\beta$.  Thus $f(t)\geq \frac{f(t+1)^2}{f(t+2)}$, which is log-convexity. 

The converse implication is obtained by reversing the steps in the proof we just finished.

Now we show that \eqref{thm1it1} is equivalent to \eqref{thm1it3}. We have seen that~\eqref{thm1it1} is equivalent to log-convexity of $f$. We show that $\log f(t)$ is monotone decreasing and convex holds if and only if it is the pointwise maximum of a collection $A$ of decreasing affine functions; the argument is entirely standard.  To see this, let $g$ be a monotone decreasing and convex function, and let $t^*$ be a time period. It is sufficient to show that there is a decreasing, affine function $h$ for which $g \geq h$, and $g(t^*)=h(t^*)$. Suppose that $t^*\geq 1$. Observe that by convexity and decreasingness, $g(t^*)-g(t^*-1) \leq g(t^*+1)-g(t^*) < 0$.  Define $h(t) = g(t^*)+(g(t^*)-g(t^*-1))(t-t^*)$.  Observe that $h(t^*)=g(t^*)$ and that $h$ is decreasing and affine.  Suppose that $t \geq t^*$.  Then a simple inductive argument establishes $g(t)-g(t-1)\geq g(t^*)-g(t^*-1)$, from which we conclude that for any $s \geq t^*$, $h(s)=g(t^*)+(g(t^*)-g(t^*-1))(s-t^*)=g(t^*)+\sum_{t=t^*+1}^s (g(t^*)-g(t^*-1))\leq g(t^*)+\sum_{t=t^*+1}^s (g(t)-g(t-1))=g(s)$.  A symmetric argument demonstrates the inequality for $t \leq t^*$; finally, for $t^*=0$ it is sufficient to choose the slope of the relevant $h$ function to be $g(1)-g(0)$.

Because there are a countable number of time periods, we may take $A$ to be at most countable.

Each element of $A$ is of the form $t\mapsto a - d t$, and hence identified with a pair $(a,d)$ of scalars with $d>0$.

Each $t$ can be associated with a member of $A$.

Consider then a parimutuel economy with $N =\Na$, which is countable, and for which each $i\in N$ is associated with $(a_i,d_i)\in A$ for which $\log f(i)=a_i-id_i$ and $\log f(t) \geq a_i -td_i$.  Then, $i\in N$ has preferences $\succeq_{(a_i,d_i)}$ associated with the stationary discount factor $f_{(a_i,d_i)}(t)=(e^{-d_i})^t$.  Let $x^*_i(i)=1$ and zero otherwise.  Let $\income_i = f(i)$.

Observe that for each $i$, $\sum_t f(t) x_{i}(t)=f(i)=w_i$.  Next, let $y\in \ell^{\infty}_+$ so that $\sum_t f(t) y(t) \leq w_i$.  Then since $f(t)\geq e^{a_i}(e^{-d_i})^t$, it follows that $\sum_t e^{a_i}(e^{-d_i})^t y(t)\leq \sum_t f(t) y(t)\leq w_i$.  Finally, $\sum_t e^{a_i}(e^{-d_i})^t x^*_i(t)=e^{a_i}(e^{-d_i})^i=f(i)=w_i$.  So $x_i^*$ is feasible and maximizes agent $i$'s utility.

The converse, that any parimutuel equilibrium prices display log-convexity, proceeds as follows.  Consider any parimutuel equilibrium $(p^*,x^*)$ for an economy of agents with stationary discount factors.  Let $f=p^*$, and let $t,t+1,t+2$ and let $j\in N$ for which $x^*_j(t+1)>0$.  Then $\frac{\da_j^{t+1}}{f(t+1)}\geq \frac{\da_j^t}{f(t)}$, which implies $\da_j \geq \frac{f(t+1)}{f(t)}$.  And $\frac{\da_j^{t+1}}{f(t+1)}\geq \frac{\da_j^{t+2}}{f(t+2)}$, which implies $\da_j \leq \frac{f(t+2)}{f(t+1)}$.  Conclude $\frac{f(t+2)}{f(t+1)}\geq \frac{f(t+1)}{f(t)}$.

\begin{comment}
Then:
\begin{align*}
\sum_{i} \income_{i} & = \sum_i \sum_s (e^{-d_i})^s x^*_{i,s} e^{a_i} \\
& = \sum_i (e^{-d_i})^ie^{a_i} \\
& = \sum_i f(i)=1
\end{align*}
\end{comment}

\begin{comment}
Consider then a parimutuel economy with $I=A$, countable, and where each $(a,d)\in A$ has preferences $\succeq_{(a,d)}$ represented by the stationary discount factor $(e^{-d})^t$. Let $x^*_{(a,d),t}=1$ if $(a,d)=h(t)$, and zero otherwise. Set $v^*_i= \sum_t (e^{-d})^t x^*_{i,t}$, and let \[ 
\income_{(a,d)} = v^*_{(a,d)} (e^{-d})^{-t} f(t) = v^*_{(a,d)} e^a ,
\] for $t=h(a,d)$. Then, writing $a(t)$ and $d(t)$ when $t=h(a,d)$, \begin{align*}
\sum_{(a,d)} \income_{(a,d)} & = \sum_{(a,d)} \sum_s (e^{-d})^s x^*_{(a,d),s} e^a \\
& =  \sum_s \sum_{(a,d)} (e^{-d})^s \one_{s=h(a,d)} e^a \\
& =  \sum_s  (e^{-d(s)})^s e^{a(s)} = \sum_s f(s)=1
\end{align*}
\end{comment}

Now we turn to the equivalence between~\eqref{thm1it1} and~\eqref{thmit2.5}. 
Observe that if $f$ is positive, decreasing and satisfies non-decreasing impatience, then there exists $\gamma\in (0,1]$ with $f(t+1)/f(t)\rightarrow \gamma$.
Let $g(t)\equiv \gamma^{-t}f(t)$. Note that $g$ is  decreasing and satisfies log-convexity.  To see that it is decreasing, observe that $\frac{g(t+1)}{g(t)}= \frac{f(t+1)}{\gamma f(t)}\leq 1$ as $f(t+1)\leq \gamma f(t)$.  To see that it is  log-convex, recall that log-convexity is the same as non-decreasing impatience, and observe that $\frac{g(t+2)}{g(t+1)}=\frac{ f(t+2)}{\gamma f(t+1)}\geq \frac{ f(t+1)}{\gamma f(t)}=\frac{g(t+1)}{g(t)}$, as $f$ is log-convex.

This means that the sequence $h(t) = \log g(0)-\log g(t)$ is increasing, concave, and equals $0$ at $t=0$. We also have that $h(t+1)-h(t)=\log (g(t+1)/g(t))\to 0$, as $g(t+1)/g(t)\to 1$ by definition of $g$. By Lemma~\ref{lem:genbetadelta} there exists $\al\in\ell^1_+$ with $h(t)= \sum_s\al(s)\min\{s,t \}$. 

This tells us that \[g(t)=g(0) \prod_{s=0}^{\infty}\max\{e^{-s\alpha(s)},e^{-t\alpha(s)}\}.\] Thus, \[f(t)=\gamma^t g(t) = g(0)\prod_{s=0}^{\infty}\max\{\beta(s)^s\gamma(s)^t,(\beta(s)\gamma(s))^t\},\]
where
\begin{align*}
    \gamma(s) & = \gamma^\frac{1}{2^{s+1}} \\
    \beta(s) & = e^{-\alpha(s)}.
\end{align*}

Fix any sequence $\eta_s>0$ with $\sum_s \eta_s=1$. For each $s=0,\ldots$ the discount factor $t\mapsto \max\{\beta(s)^s\gamma(s)^t,(\beta(s)\gamma(s))^t\}$ is generalized  $\beta-\delta$ with $\beta = (\beta(s)\gamma(s))^{1/\eta_s}$ and $\delta=(\gamma(s))^{1/\eta_s}$, where the switch point is at $s$. Let $f_s$ denote this discount factor. Then we have that $f(t) = g(0)\prod_s f_s(t)^{\eta_s}$.

Conversely, it is basic algebra to see that positive and log-convex functions are preserved under both products and powers:
\begin{itemize}
\item If $f,g>0$ are log-convex, then so is $(fg)(t)=f(t)g(t)$.
\item If $f>0$ is log-convex and $\alpha>0$, then so is $f^{\alpha}(t)=(f(t))^{\alpha}$.
\end{itemize}

Clearly each generalized $\beta-\da$ discount factor is positive and log-convex.  The result then follows for countable geometric means by taking limits.

In the proof we have used the following lemma, which is an analogue of a result of \citet{blaschke1916}.

\begin{lemma}\label{lem:genbetadelta}
Suppose that $f$ satisfies
\begin{enumerate}
\item $f(0)=0$
\item $f(t) \geq 0$ for all $t>0$
\item $f$ concave, increasing, satisfies $\lim_{t\rightarrow\infty}f(t+1)-f(t)=0$.
\end{enumerate}
Then there exists $\alpha\in\ell^1_+$ for which for all $t$, $f(t) =\sum_{s=0}^{\infty}\alpha(s)\min\{s,t\}$.
\end{lemma}

\begin{proof}
Observe that if it holds that $f(t) =\sum_{s=0}^{\infty}\alpha(s)\min\{s,t\}$, then   $f(t+1)-f(t)=\sum_{s\geq t+1}\alpha(s)$. So starting from $f$ we may define, for $t\geq 1$, $\alpha(t)  =-f(t+1)+2f(t)-f(t-1) = 2[f(t) - (\frac{1}{2}f(t+1)+\frac{1}{2}f(t-1))]\geq 0$ as $f$ is concave.  Let  $\alpha(0)$ be arbitrary.

Observe that $f(t)-f(t-1)=f(t+1)-f(t)+\alpha(t)$ and by induction $f(t)-f(t-1)=f(t+k+1)-f(t+k)+\sum_{s=0}^k  \alpha(t+s)$.  Since  $\lim_{t\rightarrow\infty} f(t+k+1)-f(t+k)=0$, we can conclude that $\alpha$ is summable.
Further, this implies that 
\begin{equation}\label{eq:subtract}\begin{split}
f(t+1)-f(t)  = \sum_{s=0}^\infty \al(t+1+s) = 
\sum_{s=0}^{\infty}\alpha(s)[\min\{s,t+1\}-\min\{s,t\}].    
\end{split}
\end{equation}

Finally, the function $f^*(t)\equiv \sum_{s=0}^{\infty}\alpha(s)\min\{s,t\}$ is well defined because $f^*(t) = \sum_{s=0}^t \alpha(s)s+\sum_{s=t+1}^{\infty}t\alpha(s)$, and we have already established that $\alpha$ is summable.  Then Equation~\eqref{eq:subtract} establishes that for all $t$, $f(t)-f(t-1)=f^*(t)-f^*(t-1)$, and since $f(0)=0=f^*(0)$, we know that $f=f^*$.
\end{proof}

\subsection{Proof of Theorem~\ref{thm:aggregatorgeomean}}
The necessity of the axioms is for the most part immediate.  To see that the Pareto axiom holds, let $x,y$ for which for all $i$, $f_i(t)x \geq f_i(s)y$.  Then if $x \geq 0 > y$ or $x > 0 \geq y$, the result is obvious.  Otherwise, if $x,y>0$, then $\frac{f_i(t)}{f_i(s)}\geq \frac{y}{x}$ so that $\frac{\prod_i f_i(t)^{\alpha_i}}{\prod_i f_i(s)^{\alpha_i}}=\prod_i\left(\frac{f_i(t)}{f_i(s)}\right)^{\alpha_i} \geq \frac{y}{x}$, so that $\prod_i f_i(t)^{\alpha_i}x \geq \prod_i f_i(s)^{\alpha_i}y$; with a strict inequality if any individual inequality is strict (since each $\alpha_i > 0$).  A similar argument establishes the result when $x,y<0$.  

We turn then to showing that the axioms are sufficient. Let $\varphi$ be an aggregator that satisfies the axioms. We shall prove that it is a geometric mean.

By IIA, for all $t>0$, we may define a map $\varphi_t :[0,1]^M\rightarrow [0,1]$ via $\varphi_t(f_1(t),\ldots,f_m(t)) = \varphi(f)(t)$, where each $f_i \in \mathcal{NI}$.  To see why this map is well-defined, observe that if $f_i(t) = f'_i(t)$, then for all $x,y$, $f_i(t)x \geq f_i(0)y$ iff $f_i'(t)x\geq f'_i(0)y$ (recall that $f_i(0)=f'_i(0)=1$).  Therefore, this property holds for all $i\in M$, and consequently by IIA, we know that $\varphi(f)(t)1 = \varphi(f)(0)(\varphi(f)(t))$ iff $\varphi(f')(t)1 = \varphi(f')(0)(\varphi(f)(t))$.  Since $\varphi(f')(0)=\varphi(f)(0)=1$, conclude that $\varphi(f')(t)=\varphi(f)(t)$.

By the Pareto property, for all $t,s>0$, $\varphi_t = \varphi_s$: suppose that $f_i(t)=f_i(s)$ for all $i\in M$.  Then $(1,t)$ is ranked the same as $(1,s)$ for all agents, and therefore must be for the social ranking; so that $\varphi_t(f_1(t),\ldots,f_m(t))=\varphi_s(f_1(s),\ldots,f_m(s))$.  Write $\varphi^*$ for $\varphi_t$. Observe similarly by Pareto that $\varphi^*$ is strictly increasing in all coordinates, and that for any $x\in [0,1]$, $\varphi^*(x,\ldots,x)=x$.

Now, we want to claim that for all $a,b\in[0,1]^M$ with $a \leq b$, we have $\frac{\varphi^*(a)}{\varphi^*(b)}=\varphi^*\left(\frac{a_1}{b_1},\ldots,\frac{a_m}{b_m}\right)$.

To this end, let $f_1,\ldots,f_m\in\mathcal{NI}$ for which $f_i(1)=b_i$ and $f_i(2)=a_i$.  Observe that for all $i\in M$, $f_i^1(1) = \frac{a_i}{b_i}$.  By time consistency, $\varphi(f_1^1,\ldots,f_m^1)(1) = \varphi(f_1,\ldots,f_m)^1(1)$.  The left hand side is $\varphi^*(\frac{a_1}{b_1},\dots,\frac{a_m}{b_m})$ whereas the right hand side is $\frac{\varphi^*(a_1,\ldots,a_m)}{\varphi^*(b_1,\ldots,b_m)}$.  So indeed for all $a,b\in[0,1]^M$ with $a \leq b$, we have $\frac{\varphi^*(a)}{\varphi^*(b)}=\varphi^*\left(\frac{a_1}{b_1},\ldots,\frac{a_m}{b_m}\right)$.

\begin{comment}
By the first axiom, then $\varphi^*(f^s_1(t-s),\ldots,f^s_M(t-s)) = \varphi^*(f_1(t)/f_1(s),\ldots,f_M(t)/f_M(s))$. Similarly, $\varphi(f_1,\ldots,f_M)^s(t-s)$ equals 
$\varphi(f_1,\ldots,f_M)(t)/\varphi(f_1,\ldots,f_M)(t)$. Thus,
\begin{align*}
\varphi^*(f_1(t)/f_1(s),\ldots,f_M(t)/f_M(s)) & = \varphi(f^s_1,\ldots,f^s_M)(t-s)  \\
& =\varphi(f_1,\ldots,f_M)^s(t) \\
& = \frac{\varphi^*(f_1(t),\ldots,f_M(t))}{\varphi^*(f_1(s),\ldots,f_M(s))},
\end{align*} where the second equality uses time consistency and the rest the definition of $\varphi^*$ and the identities established in the prior paragraph.
\end{comment}

Observe that this is a form of the Cauchy functional equation. For $a ,b \in [0,1]^M$ with $a \leq b$, we have \[\frac{\varphi^*(a)}{\varphi^*(b)}=\varphi^*\left(\frac{a_1}{b_1},\ldots,\frac{a_m}{b_m}\right).\]

We can define $\psi:(-\infty,0]^M\rightarrow (-\infty,0]$ as $\psi(x_1,\ldots,x_m)=\log\varphi^*(\exp(x_1),\ldots,\exp(x_m))$.  Clearly $\psi(0,\ldots,0)=0$.  Observe then that $\psi$ satisfies $\psi(x-y)=\psi(x)-\psi(y)$ whenever $x \leq y$.  Analogously, $\psi(x-y)+\psi(y)=\psi(x)$, when $x\leq y$, which can equivalently be written as $\psi(x)+\psi(y)=\psi(x+y)$ for any $x,y\leq 0$. 

The result now follows from a standard Cauchy argument:  observe that for any $x\leq 0$ and any $q\in\mathbb{Q}_+$, we get $\psi(q x) = q \psi(x)$.  The monotonicity of $\psi$ then implies that for any $c\in\mathbb{R}_+$, $\psi(cx)=c\psi(x)$.

Define $\eta_i \equiv -\psi(-\mathbf{1}_i)>0$.  Then $\psi(x)=\psi(\sum_i (-x_i)(-\mathbf{1}_i))=\sum_i x_i \eta_i$, and $\sum_i \eta_i = -\psi(-\sum_i \mathbf{1}_i) = 1$ as $\varphi^*(x,\ldots,x)=x$. Thus $\log \varphi^*(e^{x_1},\ldots,e^{x_M}) = \sum_i x_i\eta_i,$ and hence $\varphi^*(e^{x_1},\ldots,e^{x_M}) = \prod_i (e^{x_i})^{\eta_i}$.

\subsection{Proof of Proposition~\ref{prop:join}}

\textbf{Establishing the direction given $\income_i$}

First, fix $\income_i\geq 0$ for which $\sum_i \income_i > 0$.  Without loss suppose that all $\income_i > 0$; otherwise, we may discard agents for which  $\income_i = 0$ and proceed.  Let  $(p^*,x^*)$ be a parimutuel equilibrium, and for each $i\in N$  define $E_i \equiv \{\omega:x_i(\omega)>0\}$. Note that $\da_i(E_i)>0$.

\textbf{Absolute continuity of $\da_i$ with respect to $p^*$.}

If $\da_i(E)>0$ for some $E\in \Sigma$ and $p^*(E) = 0$, then $\int x_i + \mathbf{1}_E d\da_i > \int x_i d\da_i$ yet $\int x_i + \mathbf{1}_E dp^* = \int x_i dp^*$, contradicting that $(p^*,x^*)$ is an equilibrium. 

\textbf{Equilibrium prices are  countably additive}

We first show that for any $E\in \Sigma$, if $E\subseteq E_i$ and $p^*(E)>0$, then $\int_E x^*_i dp^*>0$.  This follows  as $p^*(E)>0$ implies $\da_i(E)>0$ (otherwise $i$ could increase wealth by selling her consumption on $E$), and $\int_E x_i d\da_i>0$ by countable additivity of $\da_i$.  So $\int_E x_i dp^*=0$ is not possible when $(p^*,x^*)$ is an equilibrium, again because it would mean that $i$ can raise her utility for free. 

Now let us suppose by means of contradiction that $p^*$ is not countably additive.  Then there is some $y>0$ and sequence $\{F_n\}_{n\in\mathbb{N}}$ for which $F_{n+1}\subseteq F_n$ and $\bigcap_n F_n = \varnothing$, but $p^*(F_n)\geq y$.  

We may assume without loss that there is some $i\in N$ for which for all $n\in\Na$, $\int_{F_n} x^*_i dp^* \geq y/\abs{N}$.  This follows as equilibrium implies that  $p^*(F_n)=\sum_i \int_{F_n}x^*_i dp^*$ ($x^*_i$ is an allocation), so for each $n$ there is $i$ for which $\int_{F_n}x^*_i dp^*\geq \frac{y}{\abs{N}}$. We can just take an $i$ that appears infinitely often.

We know by Lebesgue dominated convergence, and the countable additivity of $\da_i$, that $\int_{F_n} x^*_i d
\da_i\rightarrow 0$.  So $\frac{\int_{F_n}x^*_i d\da_i}{\int_{F_n}x^*_i dp^*}\rightarrow 0$.  Pick $n$ large so that $\int_{F_n}x^*_i d\da_i < \frac{\int_{F_n}x^*_i dp^*}{p^*(\Omega)}$. Then $x^*_i -x^*_i|_{F_n} + \frac{\int_{F_n}x^*_i dp^*}{p^*(\Omega)}\mathbf{1}_{\Omega}$ is strictly preferred to $x^*_i$ for agent $i$, and costs the same as $x^*_i$.

\textbf{Establishing a property of ratios of measures}

Second, we show that if $E\subseteq E_i$ and $F\in \Sigma$ then \begin{equation}\label{eq:ratio}
    \da_i(F) p^*(E) \leq \da_i(E)p^*(F). 
\end{equation} Note that \eqref{eq:ratio} is immediate if $p^*(E)=0$ or (by absolute continuity) if $p^*(F)=0$. Then to prove~\eqref{eq:ratio} suppose, towards a contradiction, that  \[ 
\frac{\da_i(F)}{p^*(F)}> \frac{\da_i(E)}{p^*(E)}.
\] For $y>0$, let $E^y=\{\omega\in E: x^*_i(\omega)\geq y \}$. Note that $\cup_{y>0}E^y= E$, so the countable additivity of $\da_i$ and $p^*$ imply that there is $y>0$ with $\frac{\da_i(F)}{p^*(F)}> \frac{\da_i(E^y)}{p^*(E^y)}.$ Now observe that \[ 
\int [x^*_i - y \one_{E^y} + y \frac{p^*(E^y)}{p^*(F)} \one_{F}]d\da_i = \int x^*_i d\da_i + y [\frac{p^*(E^y)}{p^*(F)} \da_i(F) - \da_i(E^y)]
>\int x^*_i d\da_i,\] where $x^*_i - y \one_{E^y} + y \frac{p^*(E^y)}{p^*(F)} \one_{F}\geq 0$, while \[ 
\int [x^*_i - y \one_{E^y} + y \frac{p^*(E^y)}{p^*(F)} \one_{F}]dp^* = \int x^*_i dp^* + y [\frac{p^*(E^y)}{p^*(F)} p^*(F) - p^*(E^y)]
=\int x^*_i dp^*;\] a contradiction.

\textbf{Establishing absolute continuity of $p^*$ with respect to  each $\da_i$}

 For any $G\in\Sigma$ with $p^*(G)>0$, $\sum_i x^*_i=\one$ implies that there is $G_j\subseteq E_j\cap G$ with $p^*(G_j)>0$. Then~\eqref{eq:ratio} with $F=\Omega$ implies that $\da^*_j(G_j)>0$ which, by mutual absolute continuity of the $(\da_i)$, implies that $0<\da_i(G_j)\leq \da_i(G)$.

\textbf{Concluding this direction}

Next, define $\al_i  = \frac{p(E_i)}{\da_i(E_i)}>0$. Then \eqref{eq:ratio} implies that, for any $F\in \Sigma$, $p^*(F)\geq \al_i \da_i(F)$. It also implies that for any $F_i\subseteq E_i$ $p^*(F)= \al_i \da_i(F)$.

Finally, by $\sum_i x_i=\one$ we can find a collection $F_i\subseteq E_i$, for $i\in N$, pairwise disjoint, and with $F=\cup F_i$. Then \[ 
p^*(F) = \sum_{i\in N} \al_i \da_i(F_i).
\] It follows that $p^* = \bigvee \alpha_i \da_i$.

\textbf{Establishing the converse direction, given $\al_i$}

Conversely, suppose that $p = \bigvee \alpha_i \da_i$, for a collection $\al_i \geq 0$ and $\sum_i \al_i > 0$.  Let $\{E_i\}$ be a measurable partition of $X$ with the property that $p(F) = \al_i \da_i(F)$ for all $F\subseteq E_i$.  Choose $E_i = \varnothing$ when $\al_i = 0$.  Set $w_i=\al_i \da_i(E_i)$ and $x_i=\one_{E_i}$, so we have that $\int x_idp = p(E_i) = \al_i \da_i(E_i) = w_i$ and $\sum_i x_i = \one$.  

Finally, suppose that $g_i$ is such that $\int g_i dp \leq w_i = \al_i \da_i(E_i)$.  

First suppose that $\al_i > 0$.  Then since $p \geq \alpha_i \da_i$, we have $\int g_i dp \geq \int g_i d(\alpha_i \da_i)= \alpha_i \int g_i  d\da_i$.   Conclude that $\int g_i d\da_i \leq \da_i(E_i)=\int x_i d\da_i$. 

Suppose now that $\al_j = 0$, and suppose there is $g_j$ for which $\int g_j dp_j >0$, but $\int g_j dp = 0$.  We know that for any $E\in\Sigma$, $p(E) = \sum_{i\in N}\al_i \da_i(E_i\cap E)$.  So $\int g_j dp = \sum_i \al_i \int_{E_i} g_j d\da_i$.  Conclude that for every $i\in N$ for which $\al_i > 0$, $\int_{E_i} g_j d\da_i = 0$.  By mutual absolute continuity, this implies that $\int_{E_i}g_j d\da_j = 0$, and in particular $\int  g_j d\da_j = 0$, a contradiction. 

\subsection{Proof of Proposition~\ref{prop:uniqueness}}

First, let us suppose the economy is given, and that all $\da_i$ are mutually absolutely continuous.  As a consequence of the proof of Proposition~\ref{prop:join}, all of $\{\da_i\}$ and $p,\bar{p}$ are mutually absolutely continuous.  From here, we pick a probability measure $\mu$ with respect to which all measures are mutually absolutely continuous, and with a slight abuse of notation, refer to the Radon Nikodym derivative of any measure $\nu$ with respect to $\mu$ as $\nu\in L^1(\Omega,\mu)$.  

All relevant statements below are understood to hold $\mu$-almost everywhere, without further mention.

Now, as a first point, by Proposition~\ref{prop:join}, we have the existence of $\al_i$ and $\bar{\al}_i$, for each equilibrium.  

It is easy to see that for any $\omega\in\Omega$ and any $i$, if $x_i(\omega)>0$, then $p(\omega)=\alpha_i \da_i(\omega)$, so that 
$p(\omega)x_i(\omega)\frac{1}{\alpha_i}p(\omega)=p(\omega)x_i(\omega)\da_i(\omega)$.  And since $\bar{\alpha}_i \da_i(\omega) \leq \bar p(\omega)$, we conclude that $p(\omega)x_i(\omega)\da_i(\omega) \leq p(\omega)x_i(\omega)\frac{1}{\bar{\al}_i}\bar{p}(\omega)$.  Consequently:
\[p(\omega)x_i(\omega)\frac{1}{\al_i}p(\omega)\leq p(\omega)x_i(\omega)\frac{1}{\bar{\al}_i}\bar{p}(\omega).\]  Symmetrically, 
\[\bar{p}(\omega)\bar{x}_i(\omega)\frac{1}{\bar{\al}_i}\bar{p}(\omega)\leq \bar{p}(\omega)\bar{x}_i(\omega)\frac{1}{\al_i}p(\omega).\]

By mutual absolute continuity, and by multiplying the two inequalities pointwise, we have that for every $(\omega,\w')\in\Omega\times\Omega$,
$p(\omega)x_i(\omega)\bar{p}(\w')\bar{x}_i(\w')   p(\omega)\bar{p}(\w')\leq p(\omega)x_i(\omega)\bar{p}(\w')\bar{x}_i(\w')    \bar{p}(\omega)p(\w')$.

Hence, since these densities are $\mu$-almost everywhere strictly positive, \[ 
p(\omega)x_i(\omega)\bar{p}(\w')\bar{x}_i(\w') \frac{\bar p(\w')}{p(\w')}
\leq p(\omega)x_i(\omega)\bar{p}(\w')\bar{x}_i(\w') \frac{\bar p(\w)}{p(\w)}
\]

So, integrating with respect to the product measure $\mu\times \mu$ on $\W\times \W$ we obtain that \[ 
\income_i \int \bar{p}(\w')\bar{x}_i(\w') \frac{\bar p(\w')}{p(\w')}d\mu(\w')
\leq \income_i\int p(\omega)x_i(\omega) \frac{\bar p(\w)}{p(\w)}d\mu(\w)
\]

Since $\income_i>0$ and adding over $i\in N$ (which is finite), we may pass the sum inside the integral to obtain 
\[  
 \int \bar{p}(\w) (\sum_i\bar{x}_i(\w)) \frac{\bar p(\w)}{p(\w)} d\mu(\w)
\leq  \int p(\omega) (\sum_i x_i(\omega)) \frac{\bar p(\w)}{p(\w)} d\mu(\w).
\] Each of $x_i$ and $\bar x_i$ is an allocation, so \begin{equation}\label{eq:l2_v2}
     \int \bar{p}(\w)  \frac{\bar p(\w)}{p(\w)} d\mu(\w)
\leq  \int p(\omega) \frac{\bar p(\w)}{p(\w)} d\mu(\w) = 1.
\end{equation}

Observe that this inequality establishes that the function $g:\Omega\rightarrow\Re$ defined by $g(\w)=\frac{\bar{p}(\w)}{\sqrt{p(\w)}}$ satisfies $g\in L^2(\Omega,\mu)$.
Further, the function $h:\Omega\rightarrow \Re$ defined by $h(\w)=\sqrt{p(\w)}$ satisfies $h\in L^2(\Omega,\mu)$ as 
$\int h^2d\mu = \int p d\mu =1$.

By the Cauchy-Schwarz inequality:
\[ 
\left(\int g(\w)h(\w)d\mu(\w)\right)^2 \leq \int (g(\w))^2d\mu(\w) \int (h(\w))^2d\mu(\w)
\]
Observe that since $\int (h(\w))^2d\mu(\w)=1$, the right hand side of this inequality is given by $\int \frac{\bar{p}(\nu)\bar{p}(\nu)}{p(\nu)}d\mu(\nu)$, which we know by equation~\eqref{eq:l2_v2} is bounded by $1$.  On the other hand, we also know that $\int ghd\mu = \int \bar p d\mu=1$.   Conclude that $\left(\int g h d\mu \right)^2=\int g^2d\mu\int h^2d\mu$. The Cauchy-Schwarz inequality, however, only holds with equality for collinear vectors. So we may conclude that  that $g=\beta h$ almost everywhere, for some $\beta > 0$, from which we conclude using the definitions of $g$ and $h$, that $\bar{p}(\w)=\beta p(\w)$, which implies that $p = \bar{p}$ almost everywhere as each of them are densities of probability measures.  So $\bar{p}=p$.

\clearpage
\bibliographystyle{ecta}
\bibliography{parimutuel}

\end{document}